\documentclass[conference,a4paper]{IEEEtran}
\usepackage{stfloats}
\usepackage{cite}
\usepackage{amsmath,amssymb,amsfonts}
\usepackage{amsthm}
\usepackage{algorithmic}
\usepackage{hyperref}
\usepackage{graphicx}
\usepackage{textcomp}
\usepackage{xcolor}
\usepackage{float}
\usepackage{subfigure}
\usepackage{amsthm}
\usepackage{bm}
\usepackage{multicol}

\usepackage{fancyhdr}
 \usepackage{lastpage}
\usepackage{bm}

\newcommand{\boldPhi}{\boldsymbol{\Phi}}

\newcommand{\tr}{\text{trace}}
\newcommand{\Find}{\text{Find}}

\usepackage{cite}
\usepackage{amsmath,amssymb,amsfonts}
\usepackage{algorithmic}
\usepackage{graphicx}
\usepackage{textcomp}
\usepackage{xcolor}
\usepackage{float}
\usepackage{subfigure}

\usepackage{bm}
\usepackage{multicol}
\usepackage{graphicx}  

\usepackage{caption}
\usepackage{environ}         
\usepackage{etoolbox}        

 \usepackage{epsfig,endnotes}

\usepackage{floatpag,enumitem}
\usepackage{relsize}
 \usepackage{tikz}

\usepackage{multirow}
\usepackage{setspace}
\usepackage{mathrsfs}
\usepackage{bbm}
\usepackage{pifont}
\usepackage{amsfonts}

\usepackage{amsmath, amsthm}

\usepackage{tabularx}
\usepackage{listings}
\usepackage{graphicx}
\usepackage{amssymb,booktabs}
\usepackage{array}
\usepackage{cases}

 \usepackage{supertabular}
\usepackage{mathrsfs}

\usepackage{psfrag}
\usepackage{bbding}

\usepackage{float}

\usepackage{amsbsy}

\makeatletter
\providecommand{\leftsquigarrow}{%
  \mathrel{\mathpalette\reflect@squig\relax}%
}
\newcommand{\reflect@squig}[2]{%
  \reflectbox{$\m@th#1\rightsquigarrow$}%
}
\makeatother

\newenvironment{proof}{{Proof:}}{\hfill$\square$}

\usepackage{algorithm}
 \usepackage{algorithmic}

\makeatletter
\newcommand{\newalgname}[1]{%
  \renewcommand{\ALG@name}{#1}%
}

\def\centerhack#1{\hbox to 0pt{\hss\footnotesize #1\hss}}
\def\centerhackn#1{\hbox to 0pt{\hss #1\hss}}
\def\dchack#1{\vbox to 0pt{\vss{\hbox to 0pt{\hss#1\hss}}\vss}}

\usepackage{etoolbox}
\AtBeginEnvironment{align}{\setcounter{subeqn}{0}}
\newcounter{subeqn} %

\usetikzlibrary{positioning}

\newcounter{mysub}

\newtheorem{prop}{Proposition}
\usepackage{slashbox}

\def\BibTeX{{\rm B\kern-.05em{\sc i\kern-.025em b}\kern-.08em
    T\kern-.1667em\lower.7ex\hbox{E}\kern-.125emX}}
\begin{document}

\title{Intelligent Reflecting Surface Aided Network: Power Control for Physical-Layer Broadcasting\\
}

\author{Huimei Han$^{1}$, Jun Zhao$^2$, Dusit Niyato$^2$, Marco Di Renzo$^3$, Quoc-Viet Pham$^4$
\vspace{1.5mm}
\\
\fontsize{10}{10}\selectfont\itshape
$^1$College of Information Engineering, Zhejiang University of Technology,
Hangzhou, Zhejiang Province, China\\
$^2$School of Computer Science and Engineering, Nanyang Technological University, Singapore\\
$^3$Laboratory of Signals and Systems of Paris-Saclay University - CNRS, CentraleSup\'{e}lec, Univ Paris Sud, France\\
$^4$High Safety Core Technology Research Center, Inje University, South Korea\vspace{1.5mm}
\\
\fontsize{9}{9}
$^1$hmhan1215@zjut.edu.cn, $^2$\{junzhao,\,dniyato\}@ntu.edu.sg,   $^3$marco.di.renzo@gmail.com, $^4$vietpq90@gmail.com\vspace{-2mm}
}


\maketitle
\thispagestyle{fancy}
\pagestyle{fancy}
\lhead{This paper appears in the Proceedings of IEEE International Conference on Communications (ICC) 2020.\\
Please see \url{https://arxiv.org/abs/1912.03468} for the full journal version. Feel free to contact us for questions/remarks.}
\cfoot{\thepage}
\renewcommand{\headrulewidth}{0.4pt}
\renewcommand{\footrulewidth}{0pt}

\begin{abstract}
As a recently proposed idea for future wireless systems, intelligent reflecting surface (IRS) can assist   communications between entities which do not have high-quality direct channels in between. Specifically, an IRS comprises many \mbox{low-cost} passive elements, each of which reflects the incident signal by incurring a phase change so that the reflected signals  add coherently at the receiver. In this paper,  for an  IRS-aided wireless network, we study the problem of power control at the base station (BS) for physical-layer
broadcasting under quality of service (QoS) constraints at mobile users, by jointly designing the transmit beamforming  at  the BS and the phase shifts of the IRS units. Furthermore, we derive  a lower bound of the minimum transmit power at the BS to present the performance bound for optimization methods. Simulation results  show that, the transmit power at the BS approaches   the lower bound with  the increase of the number of IRS units, and is much lower  than that of the communication system without IRS.
\end{abstract}


\section{Introduction}
Benefiting from  various  advanced technologies, the fifth-generation (5G) communications achieves great improvements in  spectral efficiency, such as  massive antennas deployment at  the base station (BS) (i.e., massive multiple-input multiple-output),  non-orthogonal multiple access, millimeter wave  communications, and ultra-dense hetnets. However,  these advanced technologies will introduce great mounts of energy consumption, resulting in  the high complexity of hardware implementation, which brings great challenges for  practical implementations~\cite{b1}. For example, the ultra-dense HetNets mean that there are lots of BSs in the network, and the energy consumption  scales up  with respect to the number of BSs.  The massive antenna arrays consist of active elements  and transmit/recieve data, thus consuming energy  expensively.

\textbf{Intelligent reflecting surface.}
To improve the spectral efficiency  and reduce the energy consumption, researchers are exploring new ideas for   future wireless systems\mbox{~\cite{6gexplore1,6gexplore2,di2019smart,basar2017index}.} Among theses ideas, intelligent reflecting surface (IRS) has been considered in several studies~\cite{IRSwu2,b9,IRSindus,Multi,Los}. An IRS is   a planar array  consisting of  many reflecting and nearly passive   units. Each IRS unit is controlled by the BS remotely to change the phase of the incident signal, so that the signals at the receiver can add coherently. In other words, IRS  intelligently adjusts the propagation conditions to improve communication quality between the BS and mobile users (MUs). Since each IRS unit only reflects signals in a passive way, instead of transmitting/receiving signals in an active way,  the energy consumption is very low. In addition,  due to the characteristics of  lightweight and low profile, the IRS can be deployed  on walls/building facades, and the channel model between the BS and the IRS is  usually characterized as line of sight (LoS)~\cite{IRSwu2}.

\textbf{Our contributions.}
To address  the problem of power control  at the BS for physical layer
broadcasting under quality of service (QoS) constraints at the MUs in an IRS-aided network, we propose to  employ  the alternating optimization algorithm to  jointly design the transmit beamforming  at  the BS and the IRS units. Furthermore,  we derive  the lower bound of the minimum transmit power for the broadcast setting to present the performance bound for optimization methods.
 Simulation results  show that, for  the broadcasting transmit pattern,  the transmit power at the BS approaches   the lower bound with the increase of the number of IRS units, and   is much lower  than that of the communication system without IRS.

 \textbf{Comparing this paper with~\cite{IRSwu2,b9}.}
 Recently,  Wu and Zhang~\cite{IRSwu2,b9}  also considered downlink power control under QoS, with phase shifts of IRS units having continuous domains in~\cite{IRSwu2} and discrete domains in~\cite{b9}. The differences between our paper and~\cite{IRSwu2,b9} are twofold. First, our paper considers the  broadcast setting, while~\cite{IRSwu2,b9} are for the unicast setting. Second, under the line of sight (LoS) channel model between the BS and IRS, we analyze a lower bound of  the minimum transmit power in the general setting (i.e., an arbitrary number of antennas at the BS, an arbitrary number of IRS units, and an arbitrary number of MUs), while~\cite{IRSwu2,b9} only derived the relationship between the transmit power and the  received power in the special setting of considering single user and single-antenna BS, ignoring the channel between BS and MU, and assuming that the channel model between the BS and IRS is Rayleigh fading.

\textbf{Other related work.} For IRS-aided  wireless communications, in addition to~\cite{IRSwu2,b9} above and \cite{IRSwu2}'s conference version~\cite{IRSwu1} for downlink power control under QoS, optimizing transmit power at BS is  addressed in~\cite{Los} to maximize the minimum
 SINR among users and in~\cite{guo2019weighted,Panbroadcast}~to maximize the weighted sum of downlink
rates. An earlier draft~\cite{zhaooveriew} of our current paper summarizes problems of downlink power control under QoS in the unicast, multicast, and broadcast settings. As an updated version, this current paper adds simulation results and also derives a lower bound for  the minimum transmit power at the BS. In the absence of IRS, downlink power control under QoS for the broadcast setting  is studied in the seminal work~\cite{sidiropoulos2006transmit}.


\textbf{Organization.}
The remainder of this paper is organized as follows. Section~\ref{sec:systemmodel}  presents the system model and formulates the problem of power control under QoS. In Section~\ref{sec:Optimization}, we describe an algorithm to solve the problem. A lower bound for the minimum transmit power is  elaborated in Section~\ref{sec:analysis}. Sections~\ref{sec:simulation} and~\ref{sec:conclution} give numerical results and the conclusion respectively.

\textbf{Notation.}
We utilize italic letters, boldface
lower-case and upper-case letters to denote the scalars, vectors and matrices respectively.
$(\cdot)^T$ and $(\cdot)^H$ stand for the transpose and conjugate transpose of a matrix, respectively.  We utilize $\boldsymbol{D}_{i,j}$ and  $\boldsymbol{x}_{i}$ to stand for the element in the $i^{th}$ row and $j^{th}$  column of $\boldsymbol{D}$ and the $i^{th}$ element of $\boldsymbol{x}$ respectively.  $\mathbb{C}$ denotes the set of all complex numbers.  $\bm{{I}}$ is the identity matrix.   $\mathcal{C}\mathcal{N}( \mu  , \sigma^2 )$  denotes a circularly-symmetric complex Gaussian distribution with mean $\mu$  and variance $\sigma^2$.    Let $\| \cdot\|$ and $| \cdot |$ denote the Euclidean norm of a vector and cardinality of a set respectively. $\text{diag}(\boldsymbol{x})$ means a diagonal matrix with the element in the $i^{th}$ row and $i^{th}$  column being the $i^{th}$ element in $\boldsymbol{x}$. \text{arg}($\boldsymbol{x}$)  stands for the phase vector. $\mathbb{E}(\cdot)$ and $\textup{Var}(\cdot)$ are the expectation and variance operations, respectively. For a square  matrix $\boldsymbol{M}$, we use $\boldsymbol{M}^{-1}$ and  $\boldsymbol{M}\succeq 0$ to denote its inverse  and positive semi-definiteness. respectively.

\section{System model and  problem definition}\label{sec:systemmodel}

\subsection{System model}
\begin{figure}[!t]
  \centering
 \includegraphics[scale=0.35]{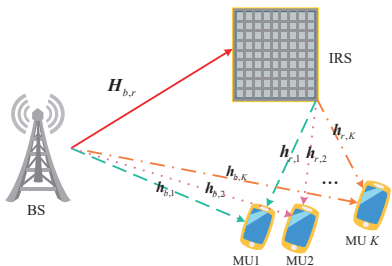}\\
  \caption{An IRS-aided communication system consisting of a base station (BS), multiple mobile users (MUs), and an IRS comprising many IRS units, where phase shifts incurred by the IRS units are remotely controlled by the BS.\vspace{-10pt}}\label{systemmodel}
\end{figure}

We consider   IRS-aided  communications in the broadcast setting, where there are a  BS with $M$ antennas and an IRS with $N$ IRS units, and  $K$ single-antenna MUs, as shown in Fig.~\ref{systemmodel}. We consider that the BS utilizes linear  transmit precoding as the beamforming vector, denoted by $\bm{w}\in \mathbb{C}^{M\times1}$, and  thus the transmitted signal at the BS is $\boldsymbol{x}= \bm{w}{s}$ where ${s}$ is  the broadcasted data. When BS broadcasts the signal $\boldsymbol{x}$, it will arrive at each MU via  indirect and direct  channels, and the received signal at each MU is the superposed signal from the two channels. More specifically, for the indirect channel, the transmitted signal $\boldsymbol{x}$  travels from the BS  to the IRS,  reflected by the IRS, and  finally travels from the IRS to these $K$ MUs. For the direct channel, the  transmitted signal $\boldsymbol{x}$  travels from  BS to these  $K$ MUs directly.

Let $\boldPhi= \text{diag}(\beta_1 e^{j \theta_1}, \ldots, \beta_N e^{j \theta_N})$ denote  the reflection coefficient matrix at the IRS, where  $\beta_n$ and $\theta_n$ denote the amplitude factor and  phase shift respectively. In this  paper, we assume that the IRS only changes  the phase of the reflected signal, i.e., $\theta_n \in [0,2\pi)$ and $\beta_n=1 $. Let $\bm{H}_{\text{b},\text{r}} \in \mathbb{C}^{N \times M} $, $\boldsymbol{h}_{\text{r},i}^H \in \mathbb{C}^{1 \times N} $,  and $\boldsymbol{h}_{\text{b},i}^H\in \mathbb{C}^{1 \times M} $ be the BS-IRS channel, IRS-$i^{th}$ MU channel, and BS-$i^{th}$ MU channel respectively. Then, the received signal at MU $i$ is given by
\begin{equation}\label{receive-data}
{y_i} = (\boldsymbol{h}_{\text{r},i}^H\boldPhi {\bm{H}_{\text{b},\text{r}}} + \boldsymbol{h}_{\text{b},i}^H)\bm{w}{s}  + {n_i}, \ \ i = 1,\ldots,K,
\end{equation}
where  ${n_i}\sim \mathcal{C}\mathcal{N}(0, \sigma^2_i) $ denotes the  additive white  Gaussian noise at MU $i$.

We assume that the broadcasted  data ${s}$ is  normalized to unit power. Then, the  signal-to-interference-plus-noise ratio (SINR) at MU $i$ can be written as
\begin{equation}\label{SINR}
\text{SINR}_i   = \frac{|\bm{h}^H_{i}(\boldPhi)\bm{w}|^2}{ \sigma^2_i},
\end{equation}
where $\boldsymbol{h}_i^H(\boldPhi)= \bm{h}^H_{\text{r},i}\boldPhi \bm{H}_{\text{b},\text{r}}+\bm{h}^H_{\text{b},i}$ means the overall downlink channel from the BS to MU $i$.
\subsection{Problem definition}
The problem of power control under  QoS for  broadcasting, is  to minimize  the  transmitted power at the BS under QoS. Note that the transmitted power at the BS is $\|\bm{w}\|^2$, and that the QoS of MU $i$ is usually characterized by its SINR.
Then, this problem can be formulated as
\begin{subequations} \label{eq-Prob-P1}
\begin{alignat}{1}
 \text{(P1):~}
\min_{\bm{w}, \boldPhi}~ &\|\bm{w}\|^2 \label{OptProb-power-broadcast-obj}   \\
~\mathrm{s.t.}~&\frac{|\bm{h}^H_{i}(\boldPhi)\bm{w}  |^2}{\sigma^2_i}\geq \gamma_i, ~\forall i = =1,\ldots,K, \label{OptProb-power-broadcast-SINR-constraint}\\
&0\le {\theta _n} < 2\pi , \ n=1,\ldots,N, \label{OptProb-power-broadcast-Phi-constraint}
\end{alignat}
\end{subequations}
 where $\gamma_i$  is the  SINR target.  Furthermore, without loss of generality, we  assume that all  MUs  have the same SINR target and    the same  noise variance, {{i.e.}, $\gamma_k=\gamma, \sigma_k^2=\sigma^2$}.

 \section{Alternating optimization algorithm } \label{sec:Optimization}


 In this section, we utilize   alternating optimization, which is  used for multivariate optimization in an alternating manner, to solve  Problem $\text{(P1)}$, as described in Algorithm \ref{Alg-Problem-P1}. More specifically, we first optimize $\bm{w}$ given $\boldPhi$, and then optimize $\boldPhi$ given $\bm{w}$, which is performed iteratively to obtain the desired $\bm{w}$ and $\boldPhi$.  In the following, we describe the details of the $j^{th}$ iteration to illustrate the alternating optimization algorithm.

\textbf{Optimizing $\bm{w}$ given $\boldPhi^{(j-1)}$.}
Given $\boldPhi^{(j-1)}$ obtained during the $(j-1)^{th}$ iteration,  Problem~(P1) becomes the conventional power control problem under QoS in the downlink broadcast channel without an IRS.
\begin{subequations} \label{eq-Prob-P1a}
\begin{alignat}{2}
 \text{(P2):~}
\min_{\bm{w}}~ &\|\bm{w}\|^2 \label{OptProb-power-broadcast-obj}   \\
~\mathrm{s.t.}~&\frac{|\bm{h}^H_{i}(\boldPhi^{(j-1)})\bm{w}  |^2}{\sigma^2}\geq \gamma, ~i =1,\ldots,K, \label{OptProb-power-broadcast-SINR-constraint}
\end{alignat}
\end{subequations}

Note that Problem (P2) is \mbox{non-convex} because of the \mbox{non-convex} constraint, which can be solved by a relaxation of Problem~(P2) based on semi-definite programming (SDP)~\cite{SDP},

\begin{subequations}
\begin{alignat}{1}
 \text{(P3):~}
\min_{\bm{X}} ~~~ &\tr(\bm{X})  \\
~\mathrm{s.t.}~&\resizebox{0.65\hsize}{!}{$\tr(\bm{X} \bm{H}_{i}(\boldPhi^{j-1}))\geq \gamma\sigma^2, ~ i =1,\ldots,K$}, \label{OptProb-power-broadcast-SINR-constraint}\\
& \boldsymbol{\bm{X}} \succeq 0,
\end{alignat}
\end{subequations}
where $\bm{X}$ and $\bm{H}_{i}(\boldPhi^{(j-1)})$ are defined as $\bm{X}:=\bm{w}\bm{w}^H$ and $\bm{H}_{i}(\boldPhi^{(j-1)}):=\bm{h}_{i}(\boldPhi^{(j-1)})\bm{h}_{i}(\boldPhi^{(j-1)})^H$ respectively.

Apparently, Problem (P3) is an  SDP, and  we can utilize the convex optimization solvers (e.g., CVX~\cite{CVX}) to solve this problem. After $\bm{X}$ is available, the Gaussian randomization~\cite{Gaussianrandom} is applied to obtain solution to Problem (P2). Note that, when utilizing the  Gaussian randomization, we can obtain many candidate solutions to  Problem (P2), and we select the one with the minimum power as the value of $\bm{w}$ during the $j^{th}$ iteration,  denoted by $\bm{w}^{(j)}$.

\textbf{Finding $\boldPhi$ given $\bm{w}^{(j)}$.} Given $\bm{w}^{(j)}$, Problem~(P1) becomes the following feasibility check problem of finding $\boldPhi$:
\begin{algorithm}[!h]
\caption{Alternating optimization to find $\bm{w}$ and $\boldsymbol{\Phi}$ for Problem~(P1).} \label{Alg-Problem-P1}
\begin{algorithmic}[1]
\setlength{\belowcaptionskip}{-0.7cm}
\setlength{\abovecaptionskip}{-0.1cm}
\STATE Initialize  $\boldsymbol{\Phi}$ as $\boldsymbol{\Phi}^{(0)}:= \text{diag}( e^{j \theta_1^{(0)}}, \ldots,  e^{j \theta_N^{(0)}})$, where $\theta_n^{(0)}$  $(n =1,2,\ldots, N)$ is chosen uniformly at random from $[0, 2\pi)$;
\STATE Initialize the iteration number $j \leftarrow 1$;
\WHILE{1}
\item[] \COMMENT{\textit{Comment: Optimizing $\bm{w}$ given $\boldPhi$:}}
\STATE {Given \hspace{-1pt}$\boldPhi$ as \hspace{-1pt}$\boldsymbol{\Phi}^{(j-1)}$, solve Problem (P3) to obtain \hspace{-1pt}$\bm{w}^{(j)}$\hspace{-1pt};\hspace{-1pt}}
\STATE \mbox{Compute the object function value $P_t^{(j)} \leftarrow \|\bm{w}^{(j)}\|^2$};
\IF{$1-\frac{P_t^{(j)}}{P_t^{(j-1)}}\le \varepsilon $}
\STATE \textbf{break}; 
\COMMENT{\textit{Comment: $\varepsilon$  controls the number of executed iterations before termination. The algorithm terminates if the relative difference between the transmit power obtained during the $j^{th}$ iteration and the $(j-1)^{th}$ iteration is no greater than $\varepsilon$.}} \label{Alg-Problem-P1-break1}
\ENDIF
\item[] \COMMENT{\textit{Comment: Finding $\boldPhi$ given $\bm{w}$:}}
\STATE Given $\bm{w}$ as $\bm{w}^{(j)}$, solve Problem~(P6) to obtain $\boldPhi^{(j)}$;
\IF{Problem~(P6) is infeasible}
\STATE \textbf{break};  \label{Alg-Problem-P1-break2}
\ENDIF
\ENDWHILE
\end{algorithmic}
\end{algorithm}

\begin{subequations} \label{eq-Prob-P1c}
\begin{alignat}{2}
\text{(P4)}: \Find~~~& \boldPhi \\
\mathrm{s.t.}~~~~~&   \frac{|\bm{h}^H_{i}(\boldPhi)\bm{w}^{(j)} |^2}{\sigma^2}\geq \gamma, \label{eq:prob-P2-SINR-constraint}\\
 &0\le {\theta _n} < 2\pi , \ n=1,...,N. \label{eq:prob-P1a-rank-constraint}
\end{alignat}
\end{subequations}
Let  $\bm{\phi}= [ e^{j \theta_1}, \ldots,  e^{j \theta_N}]^H$, $\bm{a}_{i}=\text{diag}(\bm{h}^H_{\text{r},i})\bm{H}_{\text{b},\text{r}}\bm{w}^{(j)}$, and $b_{i}=\bm{h}^H_{\text{b},i}\bm{w}^{(j)}.$  Then, the (P4) can be rewritten as
\begin{subequations} \label{eq-Prob-P1d}
\begin{alignat}{2}
\text{(P5)}: \Find~~~& \bm{\phi} \\
\mathrm{s.t.}~~~~~&    \begin{bmatrix}
\bm{\phi}^H ,~
1
\end{bmatrix} \bm{A}_i \begin{bmatrix}
\bm{\phi} \\
1
\end{bmatrix} + {b}_{i}{b}_{i}^H  \geq \gamma\sigma^2. \label{eq:prob-quadratic-constraint}  \\
 &|{\phi}_n|=1 , \ n=1,\ldots,N. \label{eq:prob-P1d-rank-constraint}
\end{alignat}
\end{subequations}
where $ \bm{A}_i= \begin{bmatrix}
\bm{a}_{i}\bm{a}_{i}^H, & \bm{a}_{i}{b}_{i}^H \\
{b}_{i}\bm{a}_{i}^H, & 0
\end{bmatrix}.$

Note that, since the constraints (\ref{eq:prob-P1d-rank-constraint}) are \mbox{non-convex},
Problem (P5) is a \mbox{non-convex} optimization problem. However,  by introducing an auxiliary variable
$t$ satisfying $|t|=1$, Problem (P5) can be converted to a  homogeneous quadratically constrained quadratic program (QCQP). Specifically, define
\begin{align}
 & \bm{v}:= t  \begin{bmatrix}
\bm{\phi}  \\
1
\end{bmatrix}=  \begin{bmatrix}
\bm{\phi}t \\
t
\end{bmatrix},  \text{and}\ \  \boldsymbol{V}:= \boldsymbol{v}  \boldsymbol{v}^H.
\end{align}

 Then, a relaxation of Problem~(P5) based on SDP  is
\begin{subequations} \label{eq-Prob-P1e}
\begin{alignat}{2}
\text{(P6)}:  \max_{\boldsymbol{V},\boldsymbol{\alpha}} \sum\nolimits_{i=1}^K   &\alpha_i\\
\mathrm{s.t.}~~~~~&   \tr(\bm{A}_i\boldsymbol{V} )  + {b}_{i} {b}_{i}^H \geq \alpha_i + \gamma \sigma^2, \label{eq:prob-P1e-prime-SINR-constraint}\\
 &\boldsymbol{V}_{n,n}  = 1,  \ \ n=1,\ldots,N+1\\
& \boldsymbol{V} \succeq 0,
 \ \alpha_i \geq 0,i=1,\ldots,K,
\end{alignat}
\end{subequations}
where  variable $\alpha_i$ can be described as MU $i$'s ``SINR residual'' in the phase shift optimization~\cite{IRSwu2}.

Similar to Problem (P2),  we can utilize the convex optimization solvers  to solve problem (P6). After $\bm{V}$ is available, the Gaussian randomization is applied to obtain many candidate solutions to  (P4), denoted by $[\boldPhi^{(j)}_1,\ldots, \boldPhi^{(j)}_c]$ where $c$ is  the number of candidate solutions. The rule of selecting one as the value of $\boldPhi$ during the $j^{th}$ iteration, denoted by  $\boldPhi^{{(j)}}$, is described as follows.

First, we define
\begin{equation} \label{eq-value-f}
f:={\min ({  {| \bm{h}_1^{H}(\boldPhi) \overline{\bm{w}} |^2},\ldots, {| \bm{h}_K^{H}(\boldPhi) \overline{\bm{w}}|^2} })},
\end{equation}
 where $\overline{\bm{w}}=\frac{\bm{w}}{\|\bm{w}\|}$ denotes the transmit beamforming direction  at the BS.
Replacing $\boldPhi$ and $\overline{\bm{w}}$  in Eq.~(\ref{eq-value-f}) with $\boldPhi^{(j-1)}$  and $\overline{\bm{w}}^{(j)}=\frac{\bm{w}^{(j)}}{\|\bm{w}^{(j)}\|}$ respectively, we can obtain the value of $f$ after optimizing $\bm{w}$ given $\boldPhi^{(j-1)}$, denoted by $f^{(j)}_{\text{o}\bm{w}}$ (the subscript ``\text{o}'' represents optimization).

 Next, after replacing $\boldPhi$ and $\bm{w}$ in Eq.~(\ref{eq-value-f}) with  $\boldPhi^{(j)}_k$ ($k=1,\ldots,c$) and $\bm{w}^{(j)}$ respectively, we can obtain the value of $f$ corresponding to $\boldPhi^{(j)}_k$, denoted by $f^{(j)}_k$. If $f^{(j)}_k$ satisfies  $f^{(j)}_k \ge f^{(j)}_{\text{o}\bm{w}}$, we incorporate it into a set $G$, and select the $\boldPhi$ corresponding to the maximum element in $G$ as  $\boldPhi^{(j)}$. We denote the  maximum element in $G$ as $f^{(j)}_{\text{o}\boldPhi}$, which is the value of $f$ after optimizing  $\boldPhi$ given $\bm{w}^{(j)}$.

\begin{prop}
The rule of selecting one as the value of $\boldPhi$ during the $j^{th}$ iteration, ensures the objective value in Problem (P2) is \mbox{non-increasing} over the iterations.
\end{prop}

\begin{proof}
Let $P_t=\|\bm{\bm{w}}\|^2$ denote  the transmit power.  Given $\boldPhi$, Problem~(P2)~can be rewritten as
\begin{equation}\label{min-power}
\begin{aligned}
&\min_{\bm{w}} P_t\\
&\mathrm{s.t.} \quad \frac{{P_t| \bm{h}_i^{H}(\boldPhi) \overline{\bm{w}} |}^2}{{\sigma}^2 }\geq \gamma \quad \forall i\in\{1,2,\ldots,K\}.
\end{aligned}
\end{equation}
 Apparently, the minimum value of $P_t$ is
\begin{equation}\label{proof-inc}
P_t= \frac{\gamma\sigma^2}    {\min (  {| \bm{h}_1^{H}(\boldPhi) \overline{\bm{w}} |^2},\ldots, {| \bm{h}_K^{H}(\boldPhi) \overline{\bm{w}} |^2} )}=\frac{\gamma\sigma^2}{f}.
\end{equation}
Note that, if the value of $f$  after optimizing $\bm{w}$ given $\boldPhi$ is \mbox{non-decreasing} over the iterations, then $P_t$ is \mbox{non-increasing} over the iterations; i.e., if $f^{(j+1)}_{\text{o}\bm{w}} \ge f^{(j)}_{\text{o}\bm{w}}$, then we have $P_t^{(j+1)} \le P_t^{(j)}$. Based on the rule of selecting one as the value of $\boldPhi$ during the $j^{th}$ iteration, it is easily to derive   $f^{{(j)}}_{\text{o}\boldPhi} \ge f^{(j)}_{\text{o}\bm{w}}$. Then, if the $\bm{w}^{
(j+1)}$ is the optimal solution to Problem (P2) during the $(j+1)^{th}$ iteration, we derive $f^{(j+1)}_{\text{o}\bm{w}} \ge f^{{(j)}}_{\text{o}\boldPhi}\ge f^{(j)}_{\text{o}\bm{w}}$. Hence, we have $P_t^{(j+1)} \le P^{(j)}_t$, which means that $P_t$ is \mbox{non-increasing} over the iterations.
\end{proof}

 \section{A lower bound for   minimum transmit power} \label{sec:analysis}
 In this section, for the IRS-aided broadcast pattern, we  derive a lower bound of the minimum transmit power.

 We assume  that  the BS-MUs and IRS-MUs channel are Rayleigh fading, and that BS-IRS channel is  LoS.
 We consider the uncorrelated Rayleigh fading channel  model for  IRS-$i^{th}$MU and BS-$i^{th}$MU; i.e.~$\bm{h}_{r,i}\sim \text{ }\mathcal{C}\mathcal{N}(0,{\beta_{r,i}^2}{\bm{I}})$,~$\bm{h}_{b,i}\sim \text{ }\mathcal{C}\mathcal{N}(0,{\beta_{b,i}^2}{\bm{I}})$, where $\beta_{r,i}^2$ and $\beta_{b,i}^2$ account for  the path loss of  IRS-MUs and BS-MUs respectively. Let ($x_{\text{BS}}$, $y_{\text{BS}}$, $z_{\text{BS}}$) and ($x_{\text{IRS}}$, $y_{\text{IRS}}$, $z_{\text{IRS}}$) be the  coordinate of BS  and  IRS respectively. Then, the channel  between the BS and IRS is given by~\cite{Los}
\begin{equation}
\bm{H}_{b,r}=\sqrt{\frac{{\beta_h}^2}{{2}}}\bm{s}\bm{g}^T,
\end{equation}
where
\begin{equation}
   \begin{array}{l}
\bm s = {[{s_1},\ldots,{s_m},\ldots,{s_M}]^{T}},\bm b = [{g_1},\ldots,{g_n},\ldots,{g_N}]^T\\
{s_m} \hspace{-1.5pt}= \hspace{-1.5pt}\exp \left( {j\frac{{2\pi }}{\lambda }{d_{{\rm{BS}}}}(m - 1)\text{sin}{\phi _{Lo{S_1}}}\text{sin}{\theta _{Lo{S_1}}}} \right)\hspace{-1.5pt},{\rm{  }}m \hspace{-1.5pt}=\hspace{-1.5pt} 1,\hspace{-1.5pt}\ldots,\hspace{-1.5pt}M\\
{g_n} \hspace{-1.5pt}=\hspace{-1.5pt} \exp \left( {j\frac{{2\pi }}{\lambda }{d_{{\rm{IRS}}}}(n - 1)\text{sin}{\phi _{Lo{S_2}}}\text{sin}{\theta _{Lo{S_2}}}} \right)\hspace{-1.5pt},{\rm{  }}n \hspace{-1.5pt}=\hspace{-1.5pt} 1,\hspace{-1.5pt}\ldots,\hspace{-1.5pt}N\\
{\theta _{Lo{S_1}}} = {\tan ^{ - 1}}\left( {\dfrac{{{d_{{\rm{BS - IRS}}}}}}{{{z_{{\rm{IRS}}}} - {z_{{\rm{BS}}}}}}} \right),{\rm{   }}{\theta _{Lo{S_2}}} = \pi  - {\theta _{Lo{S_1}}}\\
{\phi _{Lo{S_1}}} = \pi  - {\tan ^{ - 1}}\left( {\dfrac{{{y_{{\rm{IRS}}}} - {y_{{\rm{BS}}}}}}{{{x_{{\rm{IRS}}}} - {x_{{\rm{BS}}}}}}} \right),{\rm{   }}{\phi _{Lo{S_2}}} = \pi  + {\phi _{Lo{S_1}}},
\end{array}
\end{equation}
 where $\lambda $ is wavelength, $d_\text{BS}$ and $d_\text{IRS}$ are the inter-antenna separation at the BS and IRS respectively, ${\phi _{Lo{S_1}}}$ and ${\phi _{Lo{S_2}}}$ are the LoS azimuth at BS and IRS respectively,
  ${\theta _{Lo{S_1}}}$ and ${\theta _{Lo{S_2}}}$ denote the  elevation angle of departure at BS and elevation angle of arrival at IRS respectively,  $\beta_{h}^2$ accounts for  the path loss of  IRS-BS, and $d_{\text{BS-IRS}}$ represent the distance between  the BS and IRS.

Next, we present the details of deriving the lower bound of   transmit power $P_t$ with respect to  the number of IRS units $N$,  the  number of MUs $K$, and  the number of antennas $M$, considering the following  two cases of parameter settings: \mbox{1) $K=1$ and $M>1$;} \mbox{2) $K>1 $ and $M>1$.} In addition, when   discussing   the case  of $K=1$, we omit  the subscript $i$ of  $\beta_{b,i}$ and $\beta_{r,i}$ for  presentation simplicity.
%


%
\textbf{Case 1): $K=1$ and $M>1$.} Based on Eq.~(\ref{min-power}),  given $\boldPhi$, the minimum transmit power $P_t$ is $P_t=\dfrac{{\sigma}^2\gamma}{| \bm{h}_1^{H}(\boldPhi)  \overline{\bm{w}} |^2}$. Furthermore, because $\bm{h}_1^{H}(\boldPhi)$ is random variance, transmit power $P_t$ should be considered to be  the average transmit power, which is more accurately written as
\begin{equation}\label{L-hp}
P_t=\frac{{\sigma}^2\gamma}{\mathbb{E}(| \bm{h}_1^{H}(\boldPhi)  \overline{\bm{w}} |^2)}.
\end{equation}
This means that minimizing the transmit power is  equivalent to  maximizing the term
$\mathbb{E}({| \bm{h}_1^{H}(\boldPhi)  \overline{\bm{w}} |^2})$.

For  $|\bm{h_1}^\text{H}(\boldPhi) \overline{\bm{w}}|$, we have
\begin{equation}\label{L-h3}
\begin{aligned}
|  \bm{h}_1^{H}(\boldPhi) \overline{\bm{w}}  | &=| \bm{h}_{r,1}^{{H}}{\boldPhi}{\bm{H}_{b,r}} \overline{\bm{w}}+  \bm{h}_{b,1}^{{H}} \overline{\bm{w}}| \\
&\mathop  \le \limits^{(a)}  |\bm{h}_{r,1}^{{H}}{\boldPhi}{\bm{H}_{b,r}} \overline{\bm{w}}|+|\bm{h}_{b,1}^{{H}} \overline{\bm{w}}|.
\end{aligned}
\end{equation}
Based on the triangle inequality, Eq.~{\ref{L-h3}}(a) holds if and only if
$\arg(\bm{h}_{r,1}^{{H}}{\boldPhi}{\bm{H}_{b,r}} \overline{\bm{w}})=\arg(\bm{h}_{b,1}^{{H}} \overline{\bm{w}})=\varphi_0$.

Let $ {A}=|\bm{h_{r,1}}^{{H}}{\boldPhi}{\bm{H}_{b,r}} \overline{\bm{w}}|, {B}=|h_{b,1}^{{H}} \overline{\bm{w}}|$. Then,  the maximum value of  $ \mathbb{E}(|\bm{h}_1^{H}(\boldPhi)  \overline{\bm{w}}|^2)$ with respect to $\boldPhi$ and $\overline{\bm{w}}$, denoted by $Q_1$, is given by
\begin{equation}\label{L-h4}
\begin{aligned}
Q_1&=\text{max}(\mathbb{E}(|{\bm{h}_1}^{H}(\boldPhi)  \overline{\bm{w}} |^2))
=\mathbb{E}(({A}+{B})^2)\\
&=\mathbb{E}({A}^2)+2E({AB})+E({B}^2).
\end{aligned}
\end{equation}
Next, we discuss  how to derive  each term in Eq.~(\ref{L-h4}).

For $\mathbb{E}({A}^2)$, we have
\begin{equation}\label{L-2h1}
\hspace{-10pt}\left\{\begin{aligned}
\mathbb{E}(A)&\mathop =\limits^{(a)}\mathbb{E}( {\sum\nolimits_{n=1}^{N}}{|h_{r,1,n}^{H}}|| {{\sum\nolimits_{m=1}^{M}}{{H}_{b,r,m,n}}  \overline{w_m}| }  )\\
&=\mathbb{E}( {\sum\nolimits_{n=1}^{N}}{|h_{r,1,n}^{H}}| |C_n|  )\\
&\mathop =\limits^{(b)}\mathbb{E}({|h_{r,1,1}^{H}}|)(|C_1| +|C_2|+\ldots+|C_N|), \\
\mathbb{E}^2(A)  &=(|C_1| +|C_2|+\ldots+|C_N|)^2 \mathbb{E}^2({|h_{r,1,1}^{H}}|)\\
&\mathop  = \limits^{(c)} (N|C_1|)^2\frac{\beta_r^2 \pi}{4}
\mathop  \le \limits^{(d)} \frac{\pi N^2\beta_h^2 \beta_r^2M}{8},\\
\hspace{-10pt}\textup{Var}(A) &\hspace{-1pt}=\hspace{-1pt} \textup{Var}({\sum\nolimits_{n=1}^{N}}|h_{r,1,n}^{H}| |C_n|)  \mathop   \le  \limits^{(e)} \frac{\beta_r^2}{2}(2\hspace{-1pt}-\hspace{-1pt}\frac{\pi}{2})\hspace{-1pt}\times\hspace{-1pt} \frac{\beta_h^2M}{2} \hspace{-1pt}\times\hspace{-1pt} N,\\
\hspace{-10pt}\mathbb{E}(A^2)\hspace{-1pt}&=\resizebox{0.83\hsize}{!}{$\hspace{-1pt}\mathbb{E}^2(\bm{A})\hspace{-1pt}+\hspace{-1pt}\textup{Var}(A)\hspace{-1pt}\le\hspace{-1pt}\frac{\pi N^2\beta_h^2 \beta_r^2M}{8}\hspace{-1pt}+\hspace{-1pt}\frac{N\beta_r^2\beta_h^2M}{4}(2\hspace{-1pt}-\hspace{-1pt}\frac{\pi}{2})$},
\end{aligned} \right.
\end{equation}
where $ C_n={{\sum_{m=1}^{M}}{\bm{H}_{b,r,m,n}}  \overline{w_m}}$, step (a)~follows from the fact that $\arg(\bm{h_{r,1}}^{{H}}{\boldPhi}{\bm{H}_{b,r}} \overline{\bm{w}})=\varphi_0$, step~(b)~follows from the fact that $\bm{h_{r,1}}\sim \text{ }\mathcal{C}\mathcal{N}(0,{\beta_{r,1}^2}{\bm{I}})$. For step~(c), Since the element in ${\bm{H}_{b,r,m,n}}$ has the same amplitude and $[\overline{w_1},\ldots,\overline{w_M}]$ is the normalized vector,
 it is easy to derive that $|C_1|=\ldots=|C_N|$.
 ~Step~(c)~ also follows from the fact that $| {h_{r,1,1}}^{H}|$  has distribution of Rayleigh with mean  $\frac{ \beta_r {\sqrt{\pi }}  }{2}$, and steps (d)(e) follow from  the fact  that term $|C_1|^2\le \frac{{M}\beta_h^2}{2} $ and that $| {h_{r,1,1}}^{H}|$  has distribution of Rayleigh with variance $\frac{\beta_r^2}{2} (2-\frac{\pi}{2})$.

For $\mathbb{E}({B}^2)$, we have
\begin{equation}\label{L-2h2}
\left\{\begin{aligned}
\mathbb{E}({B})&\mathop =\limits^{(a)}\resizebox{0.8\hsize}{!}{$\mathbb{E}( \sum\nolimits_{m=1}^{M}| {h_{b,m}}^{H} \overline{w_m}|)=\mathbb{E}( \sum\nolimits_{m=1}^{M}| {h_{b,m}}^{H}   | | \overline{w_m}|)$}\\
  &=( | \overline{w_1}|+| \overline{w_2}|+\ldots+| \overline{w_M}|)\mathbb{E}(|h_{b,1}^{{H}}|),\\
\mathbb{E}^2(B)  &=(| \overline{w_1}|+| \overline{w_2}|+\ldots+| \overline{w_M}|)^2 \mathbb{E}^2(|h_{b,1}^{{H}}|)\\
&\mathop  \le \limits^{(b)} (M|\overline{w_1}|)^2\frac{\beta_b^2 \pi}{4}=\frac{\pi \beta_b^2 M}{4},\\
\textup{Var}(B)&=\textup{Var}(\sum\nolimits_{m=1}^{M}| {h_{b,m}}^{H} || \overline{w_m} |)\mathop  \le \limits^{(c)}\frac{\beta_b^2}{2}(2-\frac{\pi}{2}),\\
\mathbb{E}(B^2)&= \mathbb{E}^2(B)+\textup{Var}(B)
\leq \frac{\pi \beta_b^2 M}{4}+ \frac{\beta_b^2}{2}(2-\frac{\pi}{2}),
\end{aligned} \right.
\end{equation}
where step (a)~follows from the fact that $\arg(\bm{h}_{b,1}^{{H}} \overline{\bm{w}})=\varphi_0$, steps (b) and (c)~follow from the fact that $| {h_{b,1}}^{H}|$  has distribution of Rayleigh with mean  $\frac{ \beta_b {\sqrt{\pi }}  }{2}$ and variance $\frac{\beta_b^2}{2} (2-\frac{\pi}{2})$, step (b) also follows  from  the fact  that term $(| \overline{w_1}|+| \overline{w_2}|+\ldots+| \overline{w_M}|)^2$ takes  the maximum value if $| \overline{w_1}|=| \overline{w_2}|=\ldots=| \overline{w_M}|$ because of $\sum\nolimits_{m=1}^{M}| \overline{w_m}|
^2=1$.

For $\mathbb{E}(AB)$,  we have
\begin{equation}\label{L-2h3}
\begin{aligned}
&\mathbb{E}(AB)=\sqrt{\mathbb{E}^2(A)\mathbb{E}^2(B)}=\frac{N\pi \beta_r \beta_h \beta_b  M}{4\sqrt{2}}. \\
\end{aligned}
\end{equation}

Substituting Eq.~(\ref{L-2h1})-Eq.~(\ref{L-2h3}) into  Eq.~(\ref{L-h4}), we   have
\begin{equation}\label{L-2h5}
\begin{aligned}
   Q_1  &=\mathbb{E}({A}^2)+2E({AB})+E({B}^2)\\
   &=\frac{\pi N^2\beta_b^2 \beta_r^2M}{8}+\frac{N\beta_r^2\beta_h^2M}{4}(2-\frac{\pi}{2}) +\frac{N\pi \beta_r \beta_h \beta_b   M}{2\sqrt{2}} \\ &+\frac{\beta_b^2}{2}(2-\frac{\pi}{2})+\frac{\pi \beta_b^2 M}{4}.
\end{aligned}
\end{equation}

Then,  substituting Eq.~(\ref{L-2h5}) into    Eq.~(\ref{L-hp}), the lower bound of the minimum transmit power at the BS in the case of $K=1, M>1$ is obtained in Eq.~(\ref{L-011}).
\begin{figure*}
\begin{equation}\label{L-011}
\resizebox{1\hsize}{!}{
$P_t \ge P_{K=1,M>1}^L=\dfrac{{\sigma}^2\gamma}{\text{max}(\mathbb{E}(| {h_1}^{H}(\boldPhi) \overline{\bm{\omega}} |^2))}=\dfrac{{\sigma}^2\gamma}{Q_1}
=\dfrac{{\sigma}^2\gamma}{\dfrac{\pi N^2\beta_b^2 \beta_r^2M}{8}+\dfrac{N\beta_r^2\beta_h^2M}{4}(2-\dfrac{\pi}{2}) +\dfrac{N\pi \beta_r \beta_h \beta_b   M}{2\sqrt{2}}+\dfrac{\beta_b^2}{2}(2-\dfrac{\pi}{2})+\dfrac{\pi \beta_b^2 M}{4}}.$}
\end{equation}
\end{figure*}


%

\textbf{Case 2): $K>1 $ and $M>1$.} The minimum value of $P_t$  satisfying the constrains in  Eq.~(\ref{min-power}), is
\begin{equation}\label{L-3h2}
P_t= \frac{\gamma\sigma^2}{\min ( {\mathbb{E}( {| \bm{h_1}^{H}(\boldPhi)  \overline{\bm{w}} |^2})},\ldots,  {\mathbb{E}({| \bm{h_K}^{H}(\boldPhi)  \overline{\bm{w}} |^2} )})}.
\end{equation}

Based on ${\min ( {\mathbb{E}( {| \bm{h_1}^{H}(\boldPhi)  \overline{\bm{w}} |^2})},\ldots,  {\mathbb{E}({| \bm{h_K}^{H}(\boldPhi)  \overline{\bm{w}} |^2} )})} \le {\min({Q_1,Q_2,\ldots,Q_K})}$, the lower bound of the minimum transmit power at the BS is given by
\begin{equation}\label{L-002}
P_t \ge P_{K>1,M>1}^L= \frac{r\sigma^2}{\min({Q_1,Q_2,\ldots,Q_K})},
\end{equation}
where  $Q_i=\max({\mathbb{E}(|\bm{h_i}^{H}(\boldPhi) \overline{\bm{w}} |^2}))$. Eq.~(\ref{L-2h5}) only presents how to get the value of  $Q_1$, and we can use the same way to compute the other values of $Q_i$ ($i=1,\ldots,K$).

In addition, since Problem P(3) has $K$ linear constrains and $M^2$ variables, the complexity of solving  Problem P(3) is  ${O}( (K+M^2)^{3.5})$ for one iteration ~\cite{complexity}. Similarly,  the complexity of solving  Problem P(6) is  ${O}( 2K+(N+1)^2)^{3.5})$ for one iteration. Hence, the complexity of the proposed alternating  optimization is $O( (K+M^2)^{3.5})$+${O}( 2K+(N+1)^2)^{3.5})$ for one iteration. A future direction for us is to reduce the computational complexity and one potential idea is to use manifold optimization ~\cite{pan2019intelligent}.

\section{Simulation results} \label{sec:simulation}

In  this section, we utilize numerical results to validate the derived lower bound of the transmit power and the  alternating optimization algorithm. We assume that the BS with  uniform linear array of antennas is located at (0,0,0),  and that the IRS with  uniform linear array of IRS units is located at (0,50,0). The inter-antenna and inter-unit separation at BS and IRS are  half wavelength. The purpose of deploying IRS is to improve the signal strength.  To illustrate this benefiting, we assume that the MUs are uniformly located at the half circle centered at the IRS with radius 2 m as shown in Fig.~\ref{simulation-lacation}, which are the cell-edge MUs. The channel models for BS-IRS, BS-MUs and IRS-MUs are the same as we described in Section \ref{sec:analysis}, and the path loss is $\beta^2_{a,b}=C_0(d_{a,b}/D_0)^{-\alpha}$, where~$C_0=1$~m, $D_0=-30$dB, $d_{a,b}$ denotes the distance between $a$ and $b$, $\alpha$ is the path loss exponent. We set $\sigma^2=-30$ dBm, $\gamma=1 dB$,  and $\varepsilon=10^{-4}$. For BS-IRS, IRS-MUs, and BS-MUs, we set $\alpha=2,2.8,3.5$ respectively. In addition, we employ the conventional power control (i.e., without IRS, termed Without-IRS in the result figures) and  power control with random phrase shift at the IRS  (termed
Random-IRS in the result figures) as our baselines.\vspace{-10pt}

\begin{figure}[h]
\setlength{\belowcaptionskip}{-0.4cm}
\setlength{\abovecaptionskip}{-0.1cm}
  \centering
 \includegraphics[scale=0.6]{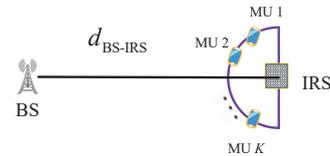}
  \caption{ The location of IRS, BS and MUs in the simulation.}\label{simulation-lacation}
\end{figure}

\begin{figure}[h]
\setlength{\belowcaptionskip}{-0.4cm}
\setlength{\abovecaptionskip}{-0.1cm}
  \centering
 \includegraphics[scale=0.33]{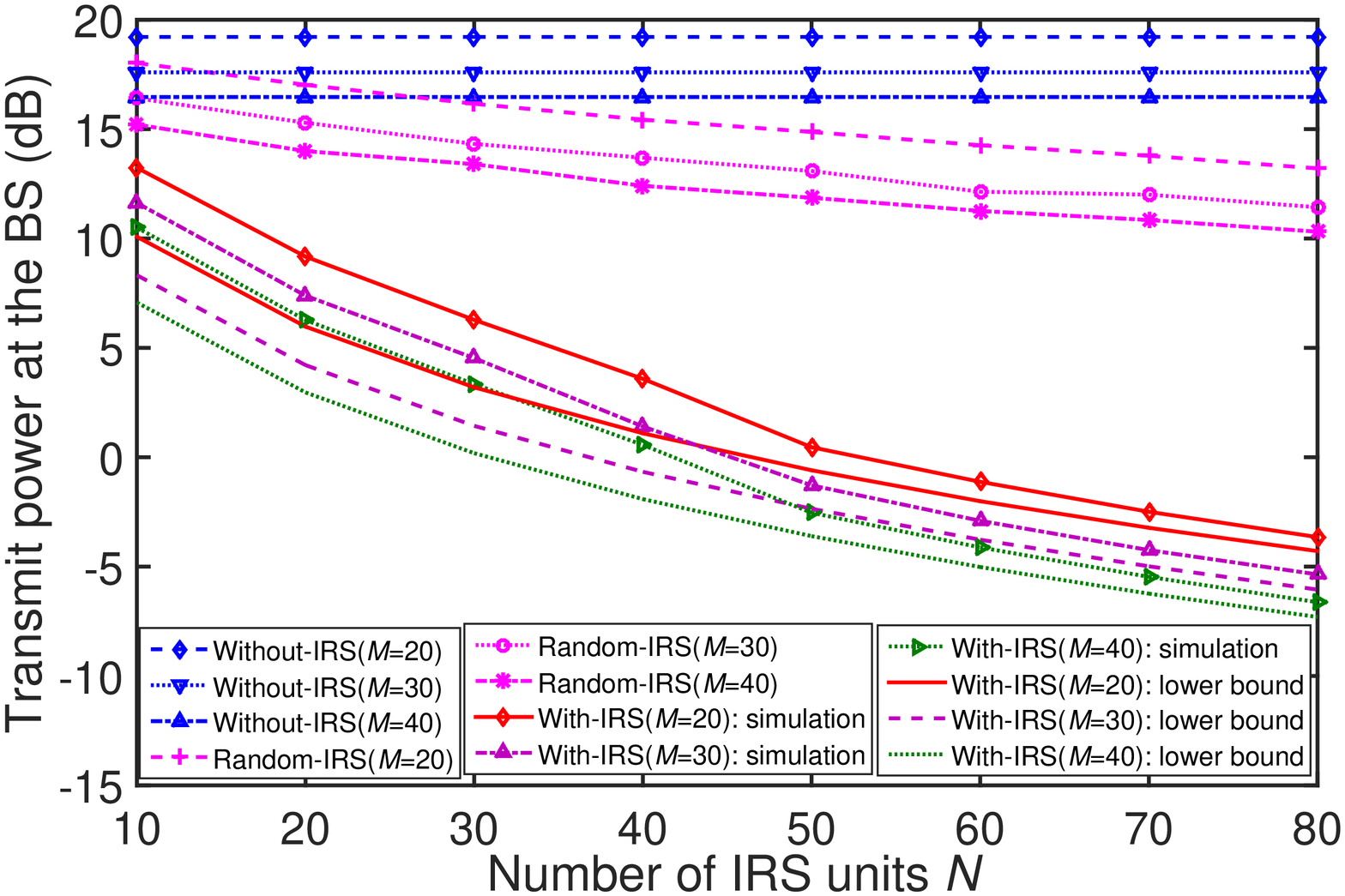}\\
  \caption{ Transmit power {\it{vs.}}  the number of IRS units ($K=1$).
  }\label{transmit-power-vs-N-K1}
\end{figure}

\begin{figure}[h]
\setlength{\belowcaptionskip}{-0.4cm}
\setlength{\abovecaptionskip}{-0.1cm}
  \centering
 \includegraphics[scale=0.33]{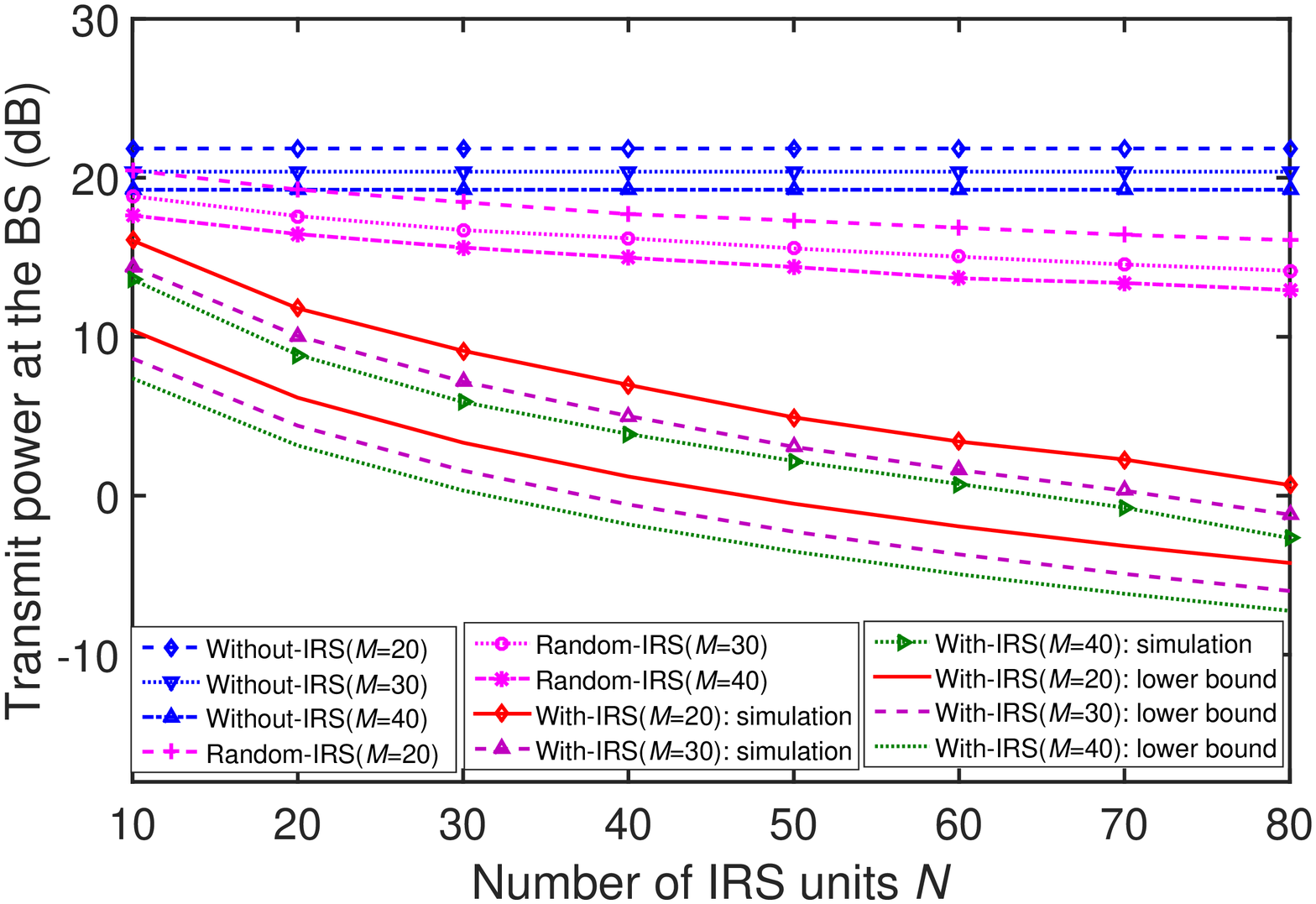}\\
  \caption{Transmit power  {\it{vs.}} the number of IRS units ($K=2$).
  }\label{transmit-power-vs-N-K2}
\end{figure}


\begin{figure}[h]
\setlength{\belowcaptionskip}{-0.4cm}
\setlength{\abovecaptionskip}{-0.1cm}
  \centering
 \includegraphics[scale=0.28]{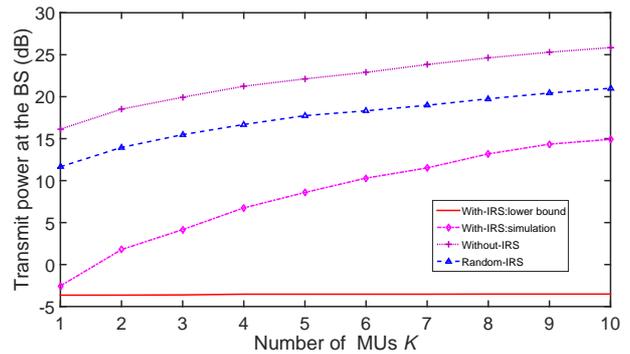}\\
  \caption{Transmit power {\it{vs.}}  the number of MUs.}\label{TK}
\end{figure}

Figs.~\ref{transmit-power-vs-N-K1}
and \ref{transmit-power-vs-N-K2}
show the variance of transmit power at BS with the number of IRS units for $M=20,30,40$ and $K=2$ respectively. We can see  that, the transmit power decreases with the increase of the number of IRS units and the number of antennas at the BS, which significantly lower than the baselines. This indicates that, deploying IRS can actually improve the signal strength and thus decrease the transmit power at the BS. Furthermore, the results from Figs.~\ref{transmit-power-vs-N-K1} and \ref{transmit-power-vs-N-K2} also show that the transmit power at the BS approaches   the lower bound  with the increase of the the number of IRS units, which coincide with the analysis results. Notice that  Fig.~\ref{transmit-power-vs-N-K1}
 is more obvious than Fig.~\ref{transmit-power-vs-N-K2}, and actually the speed of approaching   the lower bound in Fig.~\ref{transmit-power-vs-N-K2} is extremely slow.


Fig.\ref{TK} show the variance of transmit power at BS with the number of MUs $K$ ranging from 1 to 10   for $M=40, N=50$. The results show that,  with the increase of the number of MUs, the transmit power increases,  dramatically  lower than the baselines, and the gap between the transmit power and the lower bound widens up.  To obtaining a better bound which grows with $K$ is our future direction.


\section{Conclusion} \label{sec:conclution}
In this paper, we have proposed a solution to the power control under QoS for an  IRS-aided  wireless network. Specifically, we utilize the alternating optimization algorithm   to   jointly optimize  the transmit beamforming  at  the BS and the  passive IRS units at  the IRS. Furthermore, we derived  a lower bound of the minimum transmit power for the IRS-enhanced  physical layer broadcasting.
Simulation results  show that,   the transmit power at the BS approaches   the lower bound with  the number of IRS units, and   is significantly lower  than that of the communication system without IRS.


\bibliographystyle{IEEEtran}
\bibliography{ref}

\end{document}